\documentclass[letterpaper, 10 pt, journal, twoside]{IEEEtran}

\usepackage{amsmath}
\usepackage{amssymb}
\usepackage{amsthm}
\usepackage{graphicx}
\usepackage{cite}
\usepackage{algorithm}
\usepackage{algorithmic}

\newtheorem{theorem}{Theorem}
\newtheorem{lemma}[theorem]{Lemma}
\newtheorem{proposition}[theorem]{Proposition}
\newtheorem{definition}[theorem]{Definition}
\newtheorem{remark}[theorem]{Remark}
\newtheorem{corollary}[theorem]{Corollary}
\newtheorem{assumption}{Assumption}

\begin{document}

\title{The Space-Time Transform: Memory-Augmented Control Barrier
  Functions}

\author{Avinash Malik%
\thanks{Department of Electrical, Computer, and Software Engineering,
  University of Auckland, New Zealand. email: avinash.malik@auckland.ac.nz}
}

\maketitle

\begin{abstract}
  Control Barrier Functions (CBFs), their High-Order variants (HOCBFs)
  and Exponential CBFs (ECBFs) are standard geometric tools for
  enforcing nonlinear safety constraints. CBFs, and their variants,
  offer an elegant geometric framework for nonlinear safety, yet
  mathematically, they reduce to continuous-time convolutions restricted
  by zero-memory kernels. In the presence of high-frequency measurement
  noise, these memoryless operators act as improper filters, leading to
  significant control chattering and the potential loss of active
  control authority due to Quadratic Program (QP) infeasibility. To
  address this structural limitation, this paper introduces a space-time
  transform that embeds dynamic temporal filtering directly into the
  safety constraint synthesis. By designing a proper spatio-temporal
  kernel, this approach inherently attenuates high-frequency noise while
  preserving affine control authority. Crucially, we prove the robust
  forward invariance of the designed STT-CBF. Monte Carlo simulations of
  a third-order system demonstrate that the proposed framework achieves
  a $100\%$ safety rate while reducing control total variation by over
  $99\%$ compared to conventional parameterized barrier methods,
  mitigating hardware hazards and enabling reliable deployment on
  physical robotic platforms.
\end{abstract}

\section{Introduction}
\label{sec:introduction}

The enforcement of forward invariance in safety-critical systems has
been dominated by the geometry of Control Barrier Functions
(CBFs)~\cite{ames2019control}. For systems of higher relative degree,
the community standard is the High-Order CBF
(HOCBF)~\cite{xiao2019control} or the Exponential CBF
(ECBF)~\cite{nguyen2016exponential}, which rely on successive Lie
differentiation of a spatial boundary to uncover the control input.

Although this approach offers an elegant geometric interpretation, it
operates fundamentally as a continuous-time convolution with a
zero-memory kernel. Because standard barrier frameworks map spatial
boundaries to control inputs using instantaneous derivatives rather than
accumulated temporal history, they are constrained by their purely
instantaneous nature. Recent literature has identified this lack of
temporal memory, and the subsequent loss of real-time feasibility under
uncertainty, as a major issue plaguing CBF
deployment~\cite{huang2025comprehensive}. This temporal vulnerability
compounds known geometric limitations in standard barrier formulations,
such as the introduction of undesirable asymptotically stable equilibria
formally identified by Reis et al.~\cite{reis2020control}.

In this paper, we prove that this lack of memory is not merely a tuning
issue, but a structural limitation. We prove that any geometric safety
filter mapped by an improper transfer function will act as a high-pass
filter, amplifying high-frequency transient geometry into unbounded
control demands. This structurally ignores the physical plant's
bandwidth limitations, inevitably violating physical actuator limits and
violating the real-time Quadratic Program (QP).

To address this limitation, we introduce the Space-Time Transform (STT),
replacing instantaneous geometric sequences with a proper, predictive
temporal convolution. By enforcing the safety constraint through a
proper Linear Time-Invariant (LTI) dynamic extension, the transform
rigorously bounds the demanded control effort, guarantees continuous QP
feasibility under hardware saturation, and subsumes classical memoryless
CBFs into a unified spatio-temporal framework.

Our main \textbf{contributions} can be listed as follows:

\begin{itemize}
\item \textbf{Convolutional perspective of safety frameworks:} We
  establish a fundamental mathematical perspective on standard geometric
  safety filters (CBFs, HOCBFs, and ECBFs), proving that they
  structurally operate as continuous-time convolutions restricted by
  zero-memory, improper distributional kernels.
\item \textbf{Fundamental limits of instantaneous geometry:} We
  mathematically show that these memoryless spatial mappings
  structurally demand unbounded high-frequency gain. This inevitably
  leads to disturbance-induced actuator saturation and Quadratic Program
  (QP) infeasibility in higher relative degree systems.
\item \textbf{Structural failure modes of external pre-filtering:} We
  formally prove that attempting to mitigate noise via external dynamic
  pre-filtering introduces an inherent mathematical trade-off. This
  approach strictly results in either a total loss of affine control
  authority, loss of physical forward invariance, or dynamic boundary
  penetration due to temporal phase lag.
\item \textbf{The Space-Time Transform (STT) and robust forward
    invariance:} We introduce a novel spatio-temporal framework that
  replaces instantaneous geometric differentiation with a proper
  continuous-time convolution. We rigorously prove that this formulation
  bounds control effort, maintains QP feasibility under actuator limits,
  and guarantees robust physical forward invariance by transforming
  uncompensated phase lag into a strictly bounded dynamic discrepancy.
  We use Monte Carlo simulations of a third-order system to demonstrate
  that proposed framework achieves safety rates equivalent to
  state-of-the-art parameterized barrier methods, while drastically
  reducing control variance.
\end{itemize}

The remainder of this paper is structured as follows.
Section~\ref{sec:prel-geom-safety} reviews standard geometric safety
frameworks, including standard, High-Order, and Exponential CBFs.
Section~\ref{sec:conv-repr-safety} examines these classical formulations
to explore their mathematical structure as memoryless continuous-time
convolutions. Section~\ref{sec:struct-limit-inst} discusses the
challenges of these improper kernels under hardware limits and the
limitations of mitigating them via external pre-filtering. To address
these issues, Section~\ref{sec:space-time-transform} introduces the
Space-Time Transform (STT). Section~\ref{sec:synth-spat-decay} connects
physical temporal constraints to spatial geometries via temporal horizon
inversion. Section~\ref{sec:engin-design-proc} outlines a six-step
engineering design procedure for real-time execution based on this
theory. Finally, Section \ref{sec:exper-valid} provides experimental
validation against standard baselines, followed by a discussion of
related work in Section \ref{sec:related_work} and concluding remarks in
Section~\ref{sec:concl-future-work}.

\section{Preliminaries: Geometric Safety Constraints}
\label{sec:prel-geom-safety}

Consider a nonlinear affine control system of the form:
\begin{equation}
    \dot{x} = f(x) + g(x)u,
\end{equation}
where $x \in \mathcal{X} \subset \mathbb{R}^{n_x}$ is the system state and
$u \in \mathcal{U} \subset \mathbb{R}^{n_u}$ is the control input. The functions
$f$ and $g$ are assumed to be locally Lipschitz continuous. The input
space is bounded, such that $\mathcal{U} = \{u \in \mathbb{R}^{n_u} \mid \|u\|_\infty \le u_{\max}\}$.

Safety is formulated as the forward invariance of a set $\mathcal{C}$, defined as the superlevel set of a continuously differentiable function $h: \mathbb{R}^{n_x} \to \mathbb{R}$:
\begin{equation}
    \mathcal{C} = \{x \in \mathbb{R}^{n_x} \mid h(x) \ge 0\}.
\end{equation}

\paragraph{Standard Control Barrier Functions (CBFs)}
\label{sec:stand-contr-barr}

For a system where the relative degree of $h(x)$ with respect to $u$ is
$n=1$, the standard Control Barrier Function (CBF) guarantees that
$\mathcal{C}$ is forward-invariant if there exists an extended
class-$\mathcal{K}$ function $\alpha$ such that for all $x \in \mathcal{C}$:
\begin{equation}
    \sup_{u \in \mathcal{U}} \Big[ \mathcal{L}_f h(x) + \mathcal{L}_g h(x) u \Big] \ge -\alpha(h(x)). \label{eq:standard_cbf}
\end{equation}
Along the system trajectories, this maps to the strict continuous-time inequality $\dot{h}(t) + \alpha(h(t)) \ge 0$.

\paragraph{High-Order CBFs (HOCBFs)}
\label{sec:high-order-cbfs}

For systems of relative degree $n \ge 2$, $u$ does not appear in the first
derivative. The High-Order CBF (HOCBF) framework constructs a recursive
sequence of spatial constraints. Let $\psi_0(x) = h(x)$. For
$i \in \{1, \dots, n-1\}$, the state-dependent geometric sequence is:
\begin{equation}
    \psi_i(x) = \mathcal{L}_f \psi_{i-1}(x) + \alpha_i(\psi_{i-1}(x)).
\end{equation}
The synthesized active constraint enforced by the QP explicitly introduces $u$ at the $n$-th derivative. Expanding the time derivative of $\psi_{n-1}(x)$ along the system trajectories yields:
\begin{equation}
    C_{\text{HOCBF}}(x, u) = \mathcal{L}_f \psi_{n-1}(x) + \mathcal{L}_g \psi_{n-1}(x)u + \alpha_n(\psi_{n-1}(x)) \ge 0.
\end{equation}

\paragraph{Exponential CBFs (ECBFs)}
\label{sec:expon-cbfs-ecbfs}

Alternatively, the Exponential CBF (ECBF) enforces safety by
constraining the $n$-th order time derivative of $h(x)$ via a static
linear combination of its lower-order time derivatives:
\begin{equation}
    C_{\text{ECBF}}(x, u) = h^{(n)}(x, u) + K_1 h^{(n-1)}(x) + \dots + K_n h(x) \ge 0, \label{eq:ecbf_def}
\end{equation}
where
$h^{(n)}(x, u) = \mathcal{L}_f^n h(x) + \mathcal{L}_g \mathcal{L}_f^{n-1} h(x) u$. The constant gains
$K_i$ are selected such that the roots of the continuous-time
characteristic polynomial
$p(\lambda) = \lambda^n + K_1 \lambda^{n-1} + \dots + K_n = 0$ have negative real parts.

\section{Convolutional Representation of Safety Kernels}
\label{sec:conv-repr-safety}

The geometric constraints defined above are universally treated as
purely spatial boundary conditions. However, we can mathematically
deconstruct these formulations to reveal their underlying structure as
continuous-time convolution operators mapping spatial trajectory
signals.

\subsection{The CBF Kernel}
Let the instantaneous spatial boundary dynamics along state trajectory
$x(t)$ be denoted by
$z(t) = \frac{\mathrm{d}}{\mathrm{d}t}h(x(t)) + \alpha(h(x(t)))$, where
$\frac{\mathrm{d}}{\mathrm{d}t}h(x(t)) = \mathcal{L}_f h(x(t)) + \mathcal{L}_g h(x(t))u(t)$
represents the total ordinary time derivative. By the fundamental
sifting property of the Dirac delta distribution $\delta(\tau)$, any continuous
scalar signal evaluated at time $t$ can be expressed as an integral over
a forward-looking horizon $\tau \in [0, \infty)$. Substituting $z(t)$ yields:
\begin{equation}
    \int_0^\infty \delta(\tau) \Big[ \frac{\mathrm{d}}{\mathrm{d}t}h(x(t+\tau)) + \alpha\big(h(x(t+\tau))\big) \Big] \mathrm{d}\tau \ge 0. \label{eq:cbf_convolution}
\end{equation}
Equation~\eqref{eq:cbf_convolution} explicitly demonstrates that standard CBFs map spatial boundary geometry to safety constraints via a memoryless temporal kernel, $g_{\text{CBF}}(\tau) = \delta(\tau)$.

\subsection{The Unified High-Order Kernel}
\label{sec:unified-high-order}

To analyze high-degree HOCBF and ECBF constraints, we restrict the HOCBF
class-$\mathcal{K}$ functions to be linear
($\alpha_i(s) = \lambda_i s$). For a system of relative degree $n$, the unforced
spatial dynamics are governed by Lie derivatives along the drift vector
field, $\mathcal{L}_f^k h(x(t))$, for $k \in \{0, \dots, n\}$. The control input
$u(t)$ enters at order $n$ through the total time derivative:
\begin{equation}
    \frac{\mathrm{d}^n h(x(t))}{\mathrm{d}t^n} = \mathcal{L}_f^n h(x(t)) + \mathcal{L}_g \mathcal{L}_f^{n-1} h(x(t)) u(t).
\end{equation}
Under linear decay, both HOCBF and ECBF evaluate to the exact same $n$-th order affine constraint structure $C(t)$:
\begin{equation}
    C(t) = \mathcal{L}_f^n h(x(t)) + \sum_{k=0}^{n-1} C_k \mathcal{L}_f^k h(x(t)) + b(x(t))u(t) \ge 0, \label{eq:unified_polynomial}
\end{equation}
where $b(x) = \mathcal{L}_g \mathcal{L}_f^{n-1}h(x)$. For lower-order derivatives $k < n$, spatial Lie differentiation along system trajectories strictly coincides with ordinary temporal differentiation: $\mathcal{L}_f^k h(x(t)) = \frac{\mathrm{d}^k h(x(t))}{\mathrm{d}t^k}$.

\begin{theorem}[Distributional Kernel Representation]
  \label{thm:cbfrep}
  Every linear affine safety formulation (HOCBF or ECBF) of relative degree $n \ge 1$ admits a unique representation as a continuous-time temporal convolution mapping the spatial trajectory $h(t)$ through an improper distributional kernel $g(\tau)$.
\end{theorem}

\begin{proof}
  By the properties of distributions, the $k$-th ordinary temporal derivative of a scalar signal $h(t)$ is equivalent to a convolution with the $k$-th distributional derivative of the Dirac delta, $(\delta^{(k)} * h)(t) = \frac{\mathrm{d}^k h}{\mathrm{d}t^k}$. Because $u(t)$ enters exclusively through order $n$, the autonomous spatial terms $\mathcal{L}_f^k h(x(t))$ for $k \le n$ map directly to unforced temporal derivatives $\frac{\mathrm{d}^k h}{\mathrm{d}t^k}$.

  Applying this identity allows us to factor out the temporal trajectory $h(t)$ in \eqref{eq:unified_polynomial}, rewriting the autonomous component via a continuous-time convolution operator $\Omega$:
  \begin{equation}
    C(t) = \Omega[h](t) + b(x(t))u(t) \ge 0,
  \end{equation}
  where $\Omega[h](t) = \int_0^\infty g(\tau) h(t+\tau) \mathrm{d}\tau$. The effective temporal kernel $g(\tau)$ synthesizing the spatial Lie derivative operations of both HOCBF and ECBF is formulated purely from sequential distributional derivatives:
  \begin{equation}
    g(\tau) = \delta^{(n)}(\tau) + \sum_{k=0}^{n-1} C_k \delta^{(k)}(\tau). \label{eq:unified_kernel}
  \end{equation}
  Because $g(\tau)$ consists entirely of Dirac delta impulses evaluated at $\tau = 0$, it operates with zero temporal memory.
\end{proof}

\section{Structural Limitations of Instantaneous Geometry}
\label{sec:struct-limit-inst}

Having established that current safety methods structurally operate as zero-memory convolutions, we can evaluate their real-time feasibility by analyzing the frequency-domain properties of their effective kernels. 

\subsection{Memorylessness of Standard CBFs}
\label{sec:memoryl-stand-cbfs}

The standard CBF kernel $g_{\text{CBF}}(\tau) = \delta(\tau)$ is a memoryless
operator. Because it relies solely on instantaneous spatial gradients at
$t=0$, the constraint lacks an internal state or integral action.
Consequently, the mapping from the boundary trajectory to the control
input lacks a proper filtering mechanism, rendering the commanded
control effort highly sensitive to bounded measurement noise and
high-frequency disturbances.

\subsection{Fundamental Limits of Higher Order and Exponential CBFs}
\label{sec:fund-limits-high}

We now state a general impossibility result for any safety constraint
relying on an improper temporal mapping.

\begin{theorem}[Disturbance-Induced Actuator Saturation for Improper Safety Kernels]
  \label{thm:chatter}
  Consider a nonlinear affine system under a strict input saturation
  limit $\mathcal{U} = [-u_{\max}, u_{\max}]$ whose closed-loop trajectory set
  includes trajectories subjected to bounded measurement disturbances
  $n(t) \in L_\infty$ with non-zero spectral energy across arbitrarily high
  frequencies. Let the safety constraint be defined by the affine
  inequality $C(t) = \Omega[h](t) + b(x(t))u(t) \ge 0$, where
  $\Omega[h](t)$ is a continuous-time convolutional functional acting on the
  sensed state. If the Laplace-domain transfer function
  $G(s) = \mathcal{L}\{g(\tau)\}$ characterizing the kernel $g(\tau)$ is improper, then
  for any input bound $u_{\max} > 0$, there exists a valid bounded disturbance 
  trajectory $n(t) \in L_\infty$ such that the control effort $u(t)$ required to 
  satisfy the constraint $C(t) \ge 0$ strictly exceeds the admissible set 
  $\mathcal{U}$ at some time $t > 0$.
\end{theorem}

\begin{proof}
  When the constraint becomes active at the spatial boundary, we set
  $C(t) = 0$ to isolate the required nominal control effort $u^*(t)$.
  Assuming $b(x(t)) \neq 0$, the required control effort is algebraically
  given by $ u^*(t) = - \frac{1}{b(x(t))} \Omega[h](t) $. To evaluate
  feasibility under bounded control inputs $\mathcal{U}$, we examine the operator
  $\Omega$ in the frequency domain. Because $G(s)$ is improper, its numerator
  degree exceeds its denominator degree, implying:
  $ \sup_{\omega \ge 0} |G(j\omega)| = \infty $. Consequently, for any bound
  $M > 0$, there exists a frequency $\omega^* > 0$ such that $|G(j\omega^*)| > M$.

  Feasibility of the QP requires $|u^*(t)| \le u_{\max}$, which corresponds to the bound $|\Omega[h](t)| \le u_{\max} |b(x(t))|$. Choosing $M = \frac{u_{\max} |b(x(t))|}{\|n\|_{L_\infty}}$, the hypothesis guarantees the existence of a bounded disturbance $n(t)$ with spectral energy near $\omega^*$. Because the improper operator $\Omega$ provides no intrinsic high-frequency attenuation, evaluating $\Omega[h]$ under this disturbance yields $|\Omega[h](t)| > u_{\max} |b(x(t))|$, implying $|u^*(t)| > u_{\max}$. Thus, $u^*(t) \notin \mathcal{U}$, and no admissible control can satisfy $C(t) \ge 0$ at time $t$.

  Finally, by the continuity of the continuous-time convolution operator $\Omega$ and the coefficient $b(x(t))$, the strict inequality $|u^*(t)| > u_{\max}$ holds strictly under small $L_\infty$ perturbations of $x(t)$. Thus, the set of infeasible trajectories forms a non-empty open neighborhood in the trajectory space $\mathcal{C}$.
\end{proof}

With the limitation of improper kernels established, we can immediately
demonstrate why the prevailing high-relative-degree safety frameworks
are susceptible to actuator saturation.

\begin{corollary}[Actuator saturation of classical HOCBF/ECBF]
  For systems of relative degree $n \ge 3$, the improper high-order
  spatial operator in both standard HOCBF and ECBF formulations induces
  an aggressive gain divergence $\mathcal{O}(\omega^n)$, making them susceptible to
  disturbance induced actuator saturation.
\end{corollary}

\begin{proof}
  Follows directly from Theorems~\ref{thm:cbfrep} and~\ref{thm:chatter}.
\end{proof}

\begin{remark}[Extension to General Class $\mathcal{K}$ Functions]
  While Theorem~\ref{thm:chatter} and above corollaries are derived
  under linear class $\mathcal{K}$ representations to analyze frequency-domain
  operators, the actuator saturation extends to general nonlinear
  extended class $\mathcal{K}$ functions $\alpha_i(\cdot)$. Because exposing the control
  input $u$ for relative degree $n$ requires evaluating the $n$-th order
  Lie derivative $\mathcal{L}_f^n h(x)$, local linearization around any
  operational trajectory $x_0(t) \in \mathcal{C}$ yields a local Fréchet operator
  with an improper polynomial kernel degree of $n$. Thus, nonlinearly
  warping intermediate decay rates alters local boundary curvature but
  cannot eliminate the high-frequency gain divergence induced by the
  $n$-th order spatial differentiation.
\end{remark}

Faced with high-frequency noise amplified by spatial differentiation, an
intuitive and widely used engineering practice is to pre-filter state
measurements upstream of the safety filter. Applying a continuous
low-pass filter, dynamic observer, or Extended Kalman Filter
(EKF)~\cite{clark2021automatica} to smooth state estimates prior to
evaluating $h(x)$ effectively mitigates control chatter. However, as
formalized in Theorem~\ref{thm:extfilter}, naively coupling an upstream
temporal filter with a memoryless spatial barrier operator introduces a
structural trade-off between control authority, surrogate set dynamics,
and physical set safety.

\begin{theorem}[Structural Failure Modes of External Pre-Filtering]
  \label{thm:extfilter}
  Consider a nonlinear affine system of relative degree $n \ge 1$ with respect to a spatial safety boundary $h(x) \ge 0$. Let $F(s) = \frac{P(s)}{Q_m(s)}$ be an external, proper linear low-pass filter of order $m \ge 1$ with causal impulse response $f(\tau)$, generating the filtered spatial signal $\tilde{h}(t) = (f * h)(t) = \int_0^\infty f(\tau) h(t-\tau) \mathrm{d}\tau$. Applying a safety filtering operator to $\tilde{h}(t)$ exhibits the following structural failure modes:
  \begin{enumerate}
    \item \textbf{Loss of Control Authority (Truncated Differentiation):} If the safety operator is evaluated using only $n$ derivatives of $\tilde{h}(t)$, the control input $u(t)$ identically vanishes from the active constraint ($\mathcal{L}_g C(t) = 0$), rendering the QP independent of $u(t)$.
    \item \textbf{Loss of Physical Forward Invariance (Surrogate Mismatch):} If the safety operator is evaluated using $n+m$ derivatives to restore control authority, forward invariance is preserved exclusively for the surrogate set $\tilde{\mathcal{C}} = \{(x,z) \mid \tilde{h}(t) \ge 0\}$, which in general fails to guarantee forward invariance of the physical safe set $\mathcal{C} = \{x \mid h(x) \ge 0\}$.
    \item \textbf{Vulnerability to Boundary Penetration via Phase Lag:} Uncompensated causal memory $f(\tau)$ introduces temporal phase lag, which admits system trajectories where the physical state breaches the true boundary ($h(x(t)) < 0$) while the surrogate filter output remains nominally safe ($\tilde{h}(t) > 0$). Consequently, external pre-filtering can induce physical boundary penetration.
  \end{enumerate}
  Consequently, preserving control authority while maintaining forward invariance of the physical set $\mathcal{C}$ motivates the synthesis of a proper spatio-temporal filter kernel.
\end{theorem}

\begin{proof}
\textbf{Part 1: Loss of Control Authority}

Let the proper low-pass filter $F(s) = \frac{P(s)}{Q_m(s)}$, where $\deg(P) < \deg(Q_m) = m$, be represented by its minimal state-space realization:
\begin{equation}
    \dot{z}(t) = A_f z(t) + B_f h(x(t)), \quad \tilde{h}(t) = C_f z(t)
\end{equation}
where $z(t) \in \mathbb{R}^m$ and the direct feedthrough matrix $D_f = 0$. Successively differentiating the filtered output $\tilde{h}(t)$ up to the $n$-th order yields:
\begin{equation}
    \tilde{h}^{(n)}(t) = C_f A_f^n z(t) + \sum_{k=0}^{n-1} C_f A_f^{n-1-k} B_f h^{(k)}(x(t))
\end{equation}
The highest order spatial derivative present in this summation is $h^{(n-1)}(x(t))$. By definition of relative degree $n$, the control input $u(t)$ enters exclusively through $h^{(n)}(x(t)) = \mathcal{L}_f^n h(x) + \mathcal{L}_g \mathcal{L}_f^{n-1} h(x) u$. Because spatial differentiation mapped into $\tilde{h}^{(n)}(t)$ truncates at $n-1$, we have:
\begin{equation}
    \frac{\partial \tilde{h}^{(n)}(t)}{\partial u} = 0
\end{equation}
Evaluating an $n$-th order safety constraint $C(t) \ge 0$ yields
$\mathcal{L}_g C(t) = 0$, detaching the QP from $u(t)$ and causing total loss of
control authority.

\textbf{Part 2: Surrogate Mismatch}

To restore $u(t)$ to the active constraint, the augmented state $(x, z)$
must be differentiated an additional $m$ times, raising the total system
relative degree to $n+m$. Enforcing an $(n+m)$-th order CBF
guarantees forward invariance of the augmented surrogate set:
\begin{equation}
  \tilde{\mathcal{C}} = \{ (x, z) \in \mathbb{R}^{n_x+m} \mid \tilde{h}(t) \ge 0 \}
\end{equation}
However, forward invariance of $\tilde{\mathcal{C}}$ implies only that
trajectories starting in $\tilde{\mathcal{C}}$ remain in
$\tilde{\mathcal{C}}$
($\tilde{h}(t) \ge 0 \implies \tilde{h}(t+\Delta t) \ge 0$). Because
$\tilde{\mathcal{C}} \neq \mathcal{C} \times \mathbb{R}^m$, preserving forward invariance on
$\tilde{\mathcal{C}}$ does not directly ensure forward invariance of the physical
spatial set $\mathcal{C} = \{x \in \mathbb{R}^{n_x} \mid h(x) \ge 0\}$.

\textbf{Part 3: Vulnerability to Boundary Penetration}

By definition of continuous causal filtering, the surrogate boundary signal is given by the convolution:
\begin{equation}
    \tilde{h}(t) = \int_0^\infty f(\tau) h(t-\tau) \mathrm{d}\tau
\end{equation}
where $f(\tau) \ge 0$ for all $\tau \ge 0$ and $\int_0^\infty f(\tau) \mathrm{d}\tau = 1$. Because $f(\tau)$ has non-zero measure on past states ($\tau > 0$), $\tilde{h}(t)$ accumulates historical trajectory data.

Consider a physical trajectory $x(t)$ approaching the physical boundary $h(x) = 0$ from the interior with negative spatial velocity $\dot{h}(t) \le -v_{\min} < 0$. For all past times $t-\tau < t^*$, the state resided in the strict interior ($h(t-\tau) > 0$). At the instant of physical contact $t^*$, where $h(x(t^*)) = 0$, evaluating the surrogate signal yields:
\begin{equation}
    \tilde{h}(t^*) = \int_0^\infty f(\tau) h(t^* - \tau) \mathrm{d}\tau > 0
\end{equation}
Because $\tilde{h}(t^*) > 0$, the controller enforcing safety on $\tilde{\mathcal{C}}$ perceives the system as safely inside the interior, withholding maximal boundary-deflecting effort. Consequently, $x(t)$ can penetrate into the physical unsafe region ($h(x(t)) < 0$) while $\tilde{h}(t)$ remains positive, violating physical safety.

\textbf{Part 4: Motivation for Proper Spatio-Temporal Kernels}

To simultaneously restore $u(t)$ to the constraint while compensating for dynamic filter states $z(t)$, the temporal filter $F(s) = \frac{P(s)}{Q_m(s)}$ and spatial derivative operator $G_{\text{improper}}(s) = s^n + \sum_{k=0}^{n-1} C_k s^k$ can be unified into a single continuous-time operator:
\begin{equation}
    G_{\text{kernel}}(s) = F(s) \cdot G_{\text{improper}}(s) = \frac{P(s) \left( s^n + \sum_{k=0}^{n-1} C_k s^k \right)}{Q_m(s)}
\end{equation}
Let $p = \text{deg}(P)$ and $m = \text{deg}(Q_m)$. The numerator degree is $p + n$. To prevent infinite high-frequency gain that induces QP infeasibility, $G_{\text{kernel}}(s)$ must be proper, requiring $\lim_{\omega \to \infty} |G_{\text{kernel}}(j\omega)| < \infty$. This holds when $m \ge p + n$, motivating the design of a \textbf{proper spatio-temporal kernel}.
\end{proof}

\begin{remark}[The Physical Reality of Phase Lag]
  The structural loss of forward invariance established in Theorem
  \ref{thm:extfilter} is fundamentally a manifestation of phase lag. By
  definition, any proper causal filter introduces phase lag, meaning the
  surrogate signal $\tilde{h}(t)$ is reacting to a delayed, outdated
  representation of the true spatial geometry $h(x)$. Consequently, the
  safety controller is systematically blinded to the true system state
  during aggressive transients, allowing the physical plant to crash
  through the actual boundary while the lagging filtered state still
  falsely registers as ``safe.''
\end{remark}

\section{The Space-Time Transform}
\label{sec:space-time-transform}

To resolve the unbounded gain demanded by instantaneous Lie
differentiation, we generalize the safety constraint by shifting from a
purely spatial geometry to a spatio-temporal convolution. As established
in Theorem~\ref{thm:extfilter}, resolving high-frequency instability
while maintaining control authority requires replacing the improper
spatial kernel with a proper spatio-temporal kernel.

By the convolution theorem, multiplying the spatial operator
$G_{\text{improper}}(s)$ by a causal temporal filter $F(s)$ in the
Laplace domain is equivalent to convolving their respective impulse
responses in the time domain.

\begin{definition}[The Space-Time Kernel]
  \label{def:sttkernel}
  Let $G_{\text{improper}}(s) = s^n + \sum_{k=0}^{n-1} C_k s^k$ be the
  Laplace-domain representation of the memoryless spatial operator for a
  boundary $h(x)$ of relative degree $n$. Its corresponding time-domain
  impulse response is the improper distributional kernel
  $g_{\text{improper}}(\tau) = \mathcal{L}^{-1}\{G_{\text{improper}}(s)\} =
  \delta^{(n)}(\tau) + \sum_{k=0}^{n-1} C_k \delta^{(k)}(\tau)$. Let
  $f(\tau)$ be the causal impulse response of a proper low-pass filter
  $F(s) = \frac{P(s)}{Q_m(s)}$, where $p = \text{deg}(P)$ and
  $m = \text{deg}(Q_m)$. Assuming the filter's relative degree satisfies
  $m - p \ge n$, the Space-Time Kernel $g_{\text{STT}}(\tau)$ is defined as
  the time-domain convolution of the spatial operator and the temporal
  filter:
\begin{equation}
  g_{\text{STT}}(\tau) = (f * g_{\text{improper}})(\tau) = \int_0^\tau f(\sigma) g_{\text{improper}}(\tau - \sigma) \mathrm{d}\sigma
  \label{eq:2}
\end{equation}
Because $m \ge n + p$, the resulting kernel $g_{\text{STT}}(\tau)$ is devoid
of Dirac delta distributions or their derivatives, constituting a
bounded, distributed function over time.
\end{definition}

\begin{definition}[The Space-Time Transform Operator]
  \label{sec:space-time-transform-1}
  The Space-Time Transform $\Omega_{\text{STT}}$ maps the spatial boundary
  trajectory $h(x(t))$ to a feasible control constraint by integrating
  the spatial state against the Space-Time Kernel. The transformed
  active constraint is given by:
\begin{equation}
  C_{\text{STT}}(t) = \Omega_{\text{STT}}[h](t) + b(x(t))u(t) \ge 0, \label{eq:stt_constraint}
\end{equation}
where the operator evaluates the causal spatial history of the system trajectory:
\begin{equation}
  \Omega_{\text{STT}}[h](t) = \int_0^\infty g_{\text{STT}}(\tau) h(x(t-\tau)) \mathrm{d}\tau.
  \label{eq:1}
\end{equation}
\end{definition}

By embedding the temporal filter dynamics directly into the definition
of the convolution kernel $g_{\text{STT}}(\tau)$, the Space-Time Transform
enforces safety using a proper, finite-gain operator. This rigorously
bounds the demanded control effort, structurally preventing the QP
crashes associated with zero-memory HOCBFs and ECBFs, while maintaining
exact mathematical forward invariance of the augmented system state.

\begin{definition}[Linear Spatial Extension]
To compactly represent the linear spatial geometry while aligning with Exponential Control Barrier Function (ECBF) notation, let $\mathbf{h}(x) = \begin{bmatrix} h(x), & \mathcal{L}_f h(x), & \dots, & \mathcal{L}_f^{n-1} h(x) \end{bmatrix}^\top$ denote the extended spatial boundary vector. We define the linear spatial extension $\alpha_L(\cdot)$ as the inner product with the ECBF polynomial coefficients $C = [C_0, C_1, \dots, C_{n-1}]^\top$:
\begin{equation}
    \alpha_L\big(\mathbf{h}(x)\big) := \sum_{k=0}^{n-1} C_k \mathcal{L}_f^k h(x) = C^\top \mathbf{h}(x)
\end{equation}
For a system of relative degree $n=1$, this gracefully collapses to the standard class $\mathcal{K}$ representation $\alpha_L\big(h(x)\big) = C_0 h(x)$.
\end{definition}

\begin{proposition}[Lie Derivative Formulation of the Space-Time
  Transform]
  \label{sec:space-time-transform-2}
By the associativity of convolution, the Space-Time Transform operator $\Omega_{\text{STT}}[h](t)$ can be equivalently expressed by applying the improper spatial kernel to the boundary trajectory prior to temporal integration. Utilizing the linear spatial extension $\alpha_L$, this yields an alternative formulation expressed entirely in terms of the spatial Lie derivatives evaluated over the system's causal history:
\begin{equation}
    \Omega_{\text{STT}}[h](t) = \int_0^\infty f(\tau) \Big[ \mathcal{L}_f^n h\big(x(t-\tau)\big) + \alpha_L\big(\mathbf{h}(x(t-\tau))\big) \Big] \mathrm{d}\tau \label{eq:stt_lie}
\end{equation}
where $f(\tau)$ is the causal impulse response of the proper low-pass
filter $F(s)$, and $\mathcal{L}_f^n h(x)$ denotes the $n$-th order Lie derivative
along the drift vector field $f(x)$.
\end{proposition}

\begin{proof}
The Space-Time Transform evaluates the boundary via the unified spatial-temporal convolution $\Omega_{\text{STT}}[h] = (f * g_{\text{improper}}) * h$. By the associativity of the convolution operator, the purely spatial mapping can be isolated prior to temporal filtering:
\begin{equation}
    \Omega_{\text{STT}}[h] = f * (g_{\text{improper}} * h)
\end{equation}
Expanding the inner convolution with the exact definition of the improper spatial kernel analytically maps the boundary to its strict spatial derivatives. Because the spatial mapping relies on instantaneous evaluation along the vector field prior to control input, it equates to the pure Lie derivatives of the system:
\begin{align*}
  \footnotesize
    (g_{\text{improper}} * h)(t) &= \int_0^\infty \left( \delta^{(n)}(\sigma) + \sum_{k=0}^{n-1} C_k \delta^{(k)}(\sigma) \right) h(x(t-\sigma)) \mathrm{d}\sigma \nonumber \\
    &= \mathcal{L}_f^n h(x(t)) + \sum_{k=0}^{n-1} C_k \mathcal{L}_f^k h(x(t)) \nonumber \\
    &= \mathcal{L}_f^n h(x(t)) + \alpha_L\big(\mathbf{h}(x(t))\big)
\end{align*}
Substituting this spatial mapping back into the outer convolution directly yields the integral formulation over the causal past $(t-\tau)$ established in \eqref{eq:stt_lie}.
\end{proof}

\begin{remark}[Linear $\alpha_{L}$ restriction and Causal History]
  \label{rem:alpha}
  The restriction to a strictly linear extended class $\mathcal{K}$ function,
  denoted here as $\alpha_L$, is a mathematical necessity of the Space-Time
  Transform. Enforcing a bounded control effort relies entirely on
  factoring the barrier derivatives into a polynomial
  $G_{\text{improper}}(s)$ to form a proper, finite-gain transfer
  function $G_{\text{kernel}}(s)$ in the Laplace domain. This
  structurally confines the \textit{virtual dynamics} of the safety
  constraint to a Linear Time-Invariant (LTI) form, precluding the use
  of nonlinear class $\mathcal{K}$ functions. Crucially, this linear restriction
  applies only to the decay rate $\alpha_L(h)$; the underlying spatial
  geometry defined by the barrier candidate $h(x)$ can remain entirely
  nonlinear and non-convex.

  Furthermore, Equation~\eqref{eq:stt_lie} explicitly evaluates the
  spatial constraint over the \textit{causal past} $x(t-\tau)$, contrasting
  with predictive formulations that attempt to integrate over a future
  horizon $x(t+\tau)$. Evaluating a future horizon from current state data
  requires a Taylor series expansion, inherently demanding higher-order
  instantaneous spatial differentiation—the exact mechanism that
  precipitates high-frequency noise amplification and QP infeasibility.
  By restricting evaluation to the causal history, the Space-Time
  Transform successfully replaces improper predictive differentiation
  with a proper, finite-gain LTI filter mapping.
\end{remark}

\begin{remark}[Unification of Standard Memoryless Barrier Formulations]
The Space-Time Transform directly unifies and generalizes existing barrier frameworks. In the theoretical limit of infinite filter bandwidth ($\omega_P \to \infty$), the temporal impulse response collapses to the Dirac delta distribution, $f(\tau) = \delta(\tau)$. Substituting $f(\tau) = \delta(\tau)$ into \eqref{eq:stt_lie} evaluates the spatial integrand exclusively at $\tau = 0$:
\begin{equation}
    \Omega_{\text{STT}}[h](t) = \mathcal{L}_f^n h(x(t)) + \alpha_L\big(\mathbf{h}(x(t))\big)
\end{equation}
Thus, classical CBF, High-Order CBF (HOCBF), and Exponential CBF (ECBF)
formulations emerge as the memoryless, infinite-bandwidth limit of the
Space-Time Transform.
\end{remark}

\subsection{Operational Calculus of the Space-Time Transform}
\label{sec:oper-calc-space}

By mapping spatial boundary geometries into dynamic temporal states, the
Space-Time Transform ($\Omega_{\text{STT}}$) functions not merely as a
control architecture, but as a fundamental mathematical operator.

While the exhaustive functional analysis of this operator space is
deferred to future work, we formally state its fundamental algebraic
properties to establish a complete operational calculus for
spatio-temporal safety filters:

\begin{enumerate}
\item \textbf{Linearity (Superposition):} The Space-Time Transform is a
  linear functional over the space of extended boundary trajectories.
  For any spatial boundaries $h_1(x), h_2(x)$ and scalar constants
  $c_1, c_2 \in \mathbb{R}$:
    \begin{equation}
        \Omega_{\text{STT}}[c_1 h_1 + c_2 h_2](t) = c_1 \Omega_{\text{STT}}[h_1](t) + c_2 \Omega_{\text{STT}}[h_2](t)
    \end{equation}
    This superposition holds because both the underlying Lie differentiation and the temporal convolution against $f(\tau)$ are linear operators.
    
  \item \textbf{Causal Representation:} By the standard
    characterization of causal LTI systems, every valid, proper affine
    safety constraint that guarantees bounded control authority admits a
    unique impulse-response kernel in $L^1(\mathbb{R}_{\ge 0})$. Consequently, the
    class of Space-Time Operators
    $\{\Omega_{\text{STT}} : f \in L^1(\mathbb{R}_{\ge 0}), m \ge n\}$ represents a
    sufficient and closed class of proper LTI safety operators.

  \item \textbf{Uniqueness of the Synthesized Kernel:} For a given
    spatial boundary $h(x)$ of relative degree $n$, and a chosen proper
    temporal filter $F(s) = P(s)/Q_m(s)$, the resulting Space-Time
    Kernel $g_{\text{STT}}(\tau)$ is mathematically unique. The fundamental
    spatial decay operator $\alpha_L$ is algebraically locked to the unique
    polynomial expansion of $Q_m(s)$, preventing any arbitrary heuristic
    formulation of the safety geometry.
    
  \item \textbf{Invertibility (Spatial-Temporal Equivalence):}
    $\Omega_{\text{STT}}$ possesses a physical inverse mapping; to recover
    the required continuous spatial geometry $\alpha(h)$ from a designated
    temporal survival horizon $\tau$, one must invert the operator defining
    the temporal settling time.
\end{enumerate}

This final property, the formal inversion of the temporal mapping to
extract the exact spatial geometry, is the critical mechanism required
for physical hardware implementation, which we synthesize in the
following section.

\section{Synthesis of Spatial Decay via Temporal Horizon Inversion}
\label{sec:synth-spat-decay}

In the preceding section, we constructed a proper spatio-temporal kernel
to resolve the high-frequency infeasibility of classical barrier
methods. However, a critical design challenge remains: selecting the
exact coefficients $C_k$ that define the linear spatial extension $\alpha_L$.

A pervasive ambiguity in standard Control Barrier Function literature is
the conceptualization of the class $\mathcal{K}$ function
$\alpha(\cdot)$ as a rate of decay in \textit{time}. Rigorously,
$\alpha$ evaluates a decay in \textit{space}—it defines a geometric boundary
condition mapping a spatial margin $h(x)$ to a maximum allowable spatial
velocity. Because standard literature treats this spatial geometry as a
disconnected, free design parameter, the spatial coefficients $C_k$ are
often chosen via heuristic guesswork. This inadvertently desynchronizes
the spatial constraint from the physical temporal limits of the
actuators, directly causing phase lag and boundary overshoot.

The Space-Time Transform precludes this guesswork. To guarantee finite
control effort, the spatial geometry encapsulated by $C_k$ must be
analytically anchored to the physical bandwidth of the temporal filter
$F(s)$. To establish this exact mathematical bridge between space and
time, we analyze the physical time $\tau(x)$ required for a state
trajectory to traverse the spatial geometry to the boundary $h(x) = 0$.

\begin{definition}[Bhat Settling-Time Integral and the Survival Horizon]
  Following the finite-time stability framework of Bhat and
  Bernstein~\cite{bhat2000finite}, we can quantify the system's exact physical time-to-impact. Specifically, the survival horizon $\tau(x)$—the maximum allowable time for a trajectory starting at state $x$ with spatial margin $h(x) > 0$ to decay to a boundary threshold $\epsilon \ge 0$ under the continuous spatial geometry $\dot{h} = -\alpha(h)$—is defined by integrating over the spatial domain:
\begin{equation}
    \tau(x) = \int_\epsilon^{h(x)} \frac{1}{\alpha(s)} \mathrm{d}s \label{eq:bhat_integral}
\end{equation}
where $\epsilon = 0$ for finite-time stable decay functions (where $1/\alpha(s)$ is integrable at the origin), and $\epsilon \to 0^+$ represents the practical spatial settling threshold for asymptotic/exponential decay functions.
\end{definition}

\begin{remark}[Dimensionality Reduction and Gain Scheduling via $\tau(x)$]
  Mathematically, the survival horizon operates as a state-to-scalar
  mapping $\tau: \mathbb{R}^n \to \mathbb{R}_{\ge 0}$, collapsing an arbitrary
  $n$-dimensional state space into a single, physically meaningful
  scheduling variable in $\mathbb{R}$. Rather than attempting to tune or adapt
  barrier parameters across high-dimensional geometric boundaries,
  $\tau(x)$ serves as a low-dimensional scalar metric. This allows adaptive
  temporal filter bandwidths, dynamic gain schedules, and supervisory
  control switches to be parameterized entirely along a single physical
  axis: time-to-impact.
\end{remark}

\begin{lemma}[Inverse Bhat Decay Mapping under LTI Restrictions]
  \label{lem:inversebhat}
  To enforce a prescribed temporal survival horizon $\tau(h)$, the required
  spatial decay mapping $\alpha(h)$ is dictated by the inverse gradient of
  the settling-time function,
  $\alpha(h) = \left( \frac{\mathrm{d}\tau}{\mathrm{d}h} \right)^{-1}$.
  Furthermore, under the assumption that the augmented barrier dynamics
  must satisfy Linear Time-Invariant (LTI) properties (c.f.
  Remark~\ref{rem:alpha}), the unique admissible decay mapping is
  linear: $ \alpha(h) = \lambda_h h \label{eq:inverse_bhat} $, where
  $\lambda_h > 0$ defines the exact conversion rate (in rad/s) bridging the
  spatial margin to the required temporal horizon.
\end{lemma}

\begin{proof}
Applying separation of variables to the scalar boundary differential equation $\dot{h}(t) = -\alpha(h(t))$ isolates the spatial and temporal differentials:
\begin{equation}
    \mathrm{d}t = -\frac{\mathrm{d}h}{\alpha(h)}
\end{equation}
Integrating the time domain from $t=0$ to the settling time $t = \tau(x)$, and the spatial domain from the initial margin $h(x)$ to the boundary threshold $\epsilon$, yields the exact formulation in \eqref{eq:bhat_integral}. 

By the Fundamental Theorem of Calculus, differentiating the temporal settling horizon $\tau(h)$ with respect to the initial spatial margin $h$ yields the marginal time cost per unit of space:
\begin{equation}
    \frac{\mathrm{d}\tau}{\mathrm{d}h} = \frac{1}{\alpha(h)} \implies \alpha(h) = \left( \frac{\mathrm{d}\tau}{\mathrm{d}h} \right)^{-1}
\end{equation}
In general, this inversion maps to a broad class of nonlinear decay functions. However, the Space-Time Transform strictly requires the resulting spatial operator to form an LTI system to permit frequency-domain synthesis. Under this LTI restriction, the decay function is constrained to be proportional to the state. Therefore, the unique admissible mapping is strictly linear, establishing $\alpha(h) = \lambda_h h$ and completing the derivation.
\end{proof}

\begin{remark}[Worst-Case Lower Bounding of Survival Horizon]
  Because safe execution requires the differential inequality
  $\dot{h}(x) \ge -\alpha\big(h(x)\big)$, evaluating the Bhat integral at the
  boundary equality $\dot{h} = -\alpha(h)$ yields a guaranteed lower bound on
  the actual survival time: $\tau_{\text{actual}}(x) \ge \tau(x)$. Consequently,
  synthesizing $\alpha_L$ via inverse Bhat mapping guarantees that the
  temporal filter $F(s)$ retains sufficient bandwidth to process state
  transients even under the most aggressive allowable approach to the
  safety boundary.
\end{remark}

\begin{remark}[The $F(s) \rightarrow \tau \rightarrow \alpha \rightarrow C_k$ Mapping Pipeline]
  \label{rem:alphaL}
  Lemma~\ref{lem:inversebhat} formalizes the direct causal pipeline
  linking physical temporal hardware constraints to the spatial barrier
  coefficients $C_k$:
\begin{enumerate}
    \item \textbf{Actuator Bandwidth to Filter Poles ($F(s)$):} The plant's actuator bandwidth $\omega_P$ upper-bounds the bandwidth of the temporal filter $F(s)$. The dominant pole $p_{\text{dom}}$ of $F(s)$ defines the fastest allowable physical time constant of the closed-loop system.
    \item \textbf{Filter Poles to Survival Horizon ($\tau(x)$):} To prevent transient phase lag and boundary overshoot, the target temporal survival horizon must be bounded by this physical time constant, meaning the marginal time cost scales as $\frac{\mathrm{d}\tau}{\mathrm{d}h} \sim \frac{1}{|p_{\text{dom}}|}$.
    \item \textbf{Horizon Inversion to Spatial Decay ($\alpha_L$):} Inverting
      $\tau(h)$ via Lemma~\ref{lem:inversebhat} forces the fundamental
      spatial decay rate to perfectly match the dominant temporal filter
      pole: $\lambda_h = |p_{\text{dom}}|$.
\end{enumerate}
Expanding this relation across an $n$-th order relative degree system
dictates that the spatial decay operator must exactly mirror the
characteristic polynomial of the temporal filter. Thus, the required
spatial coefficients $C_k$ are derived by expanding the polynomial
$\prod_{i=1}^n (s + p_i)$, natively synchronizing the linear spatial
extension $\alpha_L\big(\mathbf{h}(x)\big)$ with the temporal memory lag of
$F(s)$.
\end{remark}

\section{Robust Forward Invariance via the Space-Time Transform}
\label{sec:robust-forward-invariance}

We now present the main theoretical result of this paper. As established
in Theorem~\ref{thm:extfilter}, naive external pre-filtering compromises
physical safety; uncompensated temporal phase lag tricks the controller
into withholding actuation, leading to uncontrolled boundary
penetration. Furthermore, strict forward invariance of the exact
physical boundary ($h(x) \ge 0$) is mathematically irreconcilable with the
time delay inherent to any strictly proper causal filter.

Rather than ignoring this physical reality, the Space-Time Transform
(STT) resolves it structurally. By utilizing the architectural
separation and operator positivity established in Lemmas
\ref{lem:relative_degree}--\ref{lem:filter_boundedness}, the STT
framework collapses the uncompensated phase lag into a strictly bounded
dynamic discrepancy. This allows us to invoke the Input-to-State Safe
Control Barrier Function (ISS-CBF) framework~\cite{kolathaya2018input}.
We prove that while causal filtering makes exact boundary tracking
impossible, STT transforms uncontrolled boundary penetration into a
formally quantifiable and bounded ultimate safe set.

The forward invariance proof relies on six explicit mathematical and
structural assumptions regarding the temporal filter $F(s)$.

\begin{assumption}[Structure as an EMA Cascade / External Positivity]
  \label{assum:ema_cascade}
  The temporal filter $F(s)$ is restricted to a cascade of $m$ first-order Exponential Moving Average (EMA) sections with unit DC gain ($F(0) = 1$):
  \begin{equation}
    F(s) = \prod_{i=1}^{m} \frac{p_i}{s+p_i}
  \end{equation}
  This guarantees that the filter's continuous-time impulse response is strictly non-negative ($f(t) \ge 0$ for all $t \ge 0$), making $F(s)$ an externally positive causal system.
\end{assumption}

\begin{assumption}[Strict Stability and Real Poles]
\label{assum:strict_stability}
All poles of $F(s)$ are strictly real and positive ($p_i > 0$ for all $i \in \{1, \dots, m\}$). This prevents oscillatory behavior in memory regularization and ensures Bounded-Input Bounded-Output (BIBO) stability with exponential memory decay.
\end{assumption}

\begin{assumption}[Minimum Filter Order]
\label{assum:min_filter_order}
The filter dimension $m$ must be greater than or equal to the relative degree $n$ of the physical spatial barrier $h(x)$ ($m \ge n$). This ensures the filter has enough temporal state capacity to synchronize with all higher-order spatial derivatives of $h(x)$.
\end{assumption}

\begin{assumption}[Pole Synchronization / Temporal Horizon Inversion]
  \label{assum:pole_sync}
  The first $n$ poles of the temporal filter $F(s)$ are explicitly matched to the roots of the spatial barrier operator $G_{\mathrm{improper}}(s)$:
  \begin{equation}
    G_{\mathrm{improper}}(s) = s^n + \sum_{k=0}^{n-1} C_k s^k = \prod_{i=1}^{n} (s + p_i)
  \end{equation}
  The coefficients $C_k$ of the linear spatial extension
  $\alpha_L\big(\mathbf{h}(x)\big) = \sum_{k=0}^{n-1} C_k \mathcal{L}_f^k h(x)$ are
  derived by expanding these shared factors. This follows from
  Section~\ref{sec:synth-spat-decay}.
\end{assumption}

\begin{assumption}[Selective Drift Regularization / Architectural Separation]
\label{assum:drift_regularization}
The operator $F(s)$ acts exclusively on the drift-induced spatial
certificate
$\eta(t) = \mathcal{L}_f^n h(x(t)) + \alpha_L\big(\mathbf{h}(x(t))\big)$. The
instantaneous control actuation effort $g(x(t))u(t)$ bypasses $F(s)$
entirely and enters the STT constraint instantaneously. This follows
from Definition~\ref{sec:space-time-transform-1}.
\end{assumption}

\begin{assumption}[Initial State Consistency]
\label{assum:initial_consistency}
The internal filter state vector $\mathbf{z}(0) = [z_1(0), \dots, z_m(0)]^\top$ is initialized consistent with the historical trajectory of $\eta(t)$, satisfying $\mathbf{z}(0) \ge 0$ whenever $x(0) \in \mathcal{C}$.
\end{assumption}

\begin{lemma}[Relative Degree 1 Control Authority]
\label{lem:relative_degree}
Consider the nonlinear affine system $\dot{x} = f(x) + g(x)u$ and
spatial boundary $h(x)$ of relative degree $n$. Under
Assumption~\ref{assum:drift_regularization} (Selective Drift
Regularization), the Space-Time Transform constraint
$C_{\mathrm{STT}}(t) = (f*\eta)(t) + g(x(t))u(t) \ge 0$ preserves a relative
degree of 1 with respect to the control input $u(t)$:
\begin{equation}
\frac{\partial C_{\mathrm{STT}}}{\partial u} = g(x(t)) \neq 0.
\end{equation}
Hence, the STT-QP controller retains direct, instantaneous control authority at every instant $t \ge 0$.
\end{lemma}

\begin{proof}
Because $F(s)$ acts exclusively on the unactuated drift certificate $\eta(t) = \mathcal{L}_f^n h(x) + \alpha_L(\mathbf{h}(x))$, the control term $g(x)u(t)$ bypasses the temporal convolution block entirely. Differentiating $C_{\mathrm{STT}}(t)$ with respect to $u(t)$ yields $g(x(t))$, identical to classical relative-degree-1 CBFs.
\end{proof}

\vspace{1em}

\begin{lemma}[Monotonicity and Positivity Preservation of the EMA Kernel]
\label{lem:positivity}
Let the temporal filter $F(s) = \prod_{i=1}^m \frac{p_i}{s+p_i}$ satisfy Assumptions A1–A2 ($p_i > 0$). Then its continuous-time impulse response $f(t) = \mathcal{L}^{-1}\{F(s)\}$ is strictly non-negative for all $t \ge 0$. Consequently, the temporal convolution operator $\Omega[y] = (f * y)(t)$ is a linear monotone operator:
\begin{equation}
a(t) \ge b(t) \implies (f * a)(t) \ge (f * b)(t), \quad \forall t \ge 0.
\end{equation}
\end{lemma}

\begin{proof}
Each first-order EMA stage has an impulse response $f_i(t) = p_i e^{-p_i t} \ge 0$ for $p_i > 0, t \ge 0$. The total impulse response $f(t) = (f_1 * f_2 * \dots * f_m)(t)$ is the dynamic convolution of non-negative functions, which is strictly non-negative ($f(t) \ge 0$). Linear integration against a non-negative kernel preserves inequalities on function spaces.
\end{proof}

\vspace{1em}

\begin{lemma}[Filter Discrepancy Bound via EMA Dynamics]
\label{lem:filter_boundedness}
Let the unactuated drift certificate $\eta(t)$ be Lipschitz continuous on a compact operating domain with bounded derivative $|\dot{\eta}(t)| \le L$, and let the first-order EMA filter satisfy:
\begin{equation}
    \dot{z}(t) = p\big(\eta(t) - z(t)\big), \quad p > 0.
\end{equation}
Define the filter discrepancy $\Delta(t) = \eta(t) - z(t)$. Under Assumption~\ref{assum:initial_consistency} (consistent initialization $\Delta(0) = 0$), the discrepancy is strictly bounded for all $t \ge 0$:
\begin{equation}
    |\Delta(t)| \le \frac{L}{p}\left(1 - e^{-pt}\right) \le \frac{L}{p} = \epsilon.
\end{equation}
Hence, the EMA filter introduces a uniformly bounded perturbation whose ultimate magnitude $\epsilon$ is inversely proportional to the filter bandwidth $p$.
\end{lemma}

\begin{proof}
Differentiating the discrepancy $\Delta(t) = \eta(t) - z(t)$ with respect to time yields:
\begin{equation}
    \dot{\Delta}(t) = \dot{\eta}(t) - \dot{z}(t) = \dot{\eta}(t) - p\Delta(t).
\end{equation}
Multiplying by the integrating factor $e^{pt}$ provides:
\begin{equation}
    \frac{\mathrm{d}}{\mathrm{d}t}\left( e^{pt}\Delta(t) \right) = e^{pt}\dot{\eta}(t).
\end{equation}
Integrating over $\tau \in [0, t]$ and substituting the initial condition $\Delta(0) = 0$ gives:
\begin{equation}
    \Delta(t) = \int_0^t e^{-p(t-\tau)}\dot{\eta}(\tau)\mathrm{d}\tau.
\end{equation}
Applying the triangle inequality and substituting the Lipschitz bound $|\dot{\eta}(\tau)| \le L$:
\begin{equation}
    |\Delta(t)| \le L \int_0^t e^{-p(t-\tau)}\mathrm{d}\tau = \frac{L}{p}\left(1 - e^{-pt}\right) \le \frac{L}{p}.
\end{equation}
Setting $\epsilon = \frac{L}{p}$ completes the proof.
\end{proof}

\begin{remark}[Cascade Generalization and Non-Ideal Initialization]
\label{rem:cascade_bound}
For an $m$-th order EMA cascade $F(s) = \prod_{i=1}^m \frac{p_i}{s+p_i}$ satisfying Assumptions~\ref{assum:ema_cascade}--\ref{assum:strict_stability}, there exists a topology-dependent constant $\kappa(F) > 0$ such that the discrepancy is bounded by the slowest pole $p_{\min} = \min_{i} p_i$:
\begin{equation}
    |\Delta(t)| \le \kappa(F) \frac{L}{p_{\min}} = \epsilon, \quad \forall t \ge 0.
\end{equation}
When identical poles are used, $\kappa(F) = 1$. Furthermore, if $\Delta(0) \neq 0$, the discrepancy includes an exponentially decaying transient $|\Delta(0)|e^{-p_{\min}t}$ that asymptotically converges to the exact same ultimate bound $\epsilon$.
\end{remark}

\begin{theorem}[Robust Forward Invariance via ISS-CBF]
\label{thm:stt_forward_invariance}
Consider the nonlinear system $\dot{x} = f(x) + g(x)u$ with safe set $\mathcal{C} = \{x : h(x) \ge 0\}$. Under continuous enforcement of the Space-Time constraint $C_{\mathrm{STT}}(t) \ge 0$, the physical spatial boundary dynamics satisfy an Input-to-State Safe (ISS) Control Barrier Function condition with explicit perturbation bound $\epsilon = \kappa(F)\frac{L}{p_{\min}}$:
\begin{equation}
    \mathcal{L}_f^n h(x) + \alpha_L(\mathbf{h}(x)) + g(x)u(t) \ge -\epsilon.
\end{equation}
Consequently, the STT-QP guarantees the robust forward invariance of an ISS-relaxed safe set $\mathcal{C}_{\mathrm{ISS}} = \{x : h(x) \ge -\delta(\epsilon)\}$ for a class $\mathcal{K}_\infty$ function $\delta$, establishing an explicit analytical trade-off between high-frequency noise attenuation ($p_{\min}$) and spatial set deformation ($\delta(\epsilon)$).
\end{theorem}

\begin{proof}
By Lemma~\ref{lem:relative_degree}, the STT-QP controller exercises instantaneous relative-degree-1 authority, actively enforcing $C_{\mathrm{STT}}(t) \ge 0$. 
The true physical spatial boundary dynamic is given by $\Gamma(t) = \eta(t) + g(x)u(t)$. The enforced STT constraint is defined as $C_{\mathrm{STT}}(t) = (f*\eta)(t) + g(x)u(t)$. Adding and subtracting $\eta(t)$ links the two conditions via the filter discrepancy:
\begin{equation}
    \Gamma(t) = C_{\mathrm{STT}}(t) + \eta(t) - (f*\eta)(t) = C_{\mathrm{STT}}(t) + \Delta(t).
\end{equation}
Enforcing $C_{\mathrm{STT}}(t) \ge 0$ yields:
\begin{equation}
    \Gamma(t) \ge \Delta(t).
\end{equation}
Applying Lemma~\ref{lem:filter_boundedness} and Remark~\ref{rem:cascade_bound} bounds the dynamic discrepancy by $\Delta(t) \ge -\epsilon$, where $\epsilon = \kappa(F)\frac{L}{p_{\min}}$. Therefore, the physical boundary dynamic satisfies:
\begin{equation}
    \Gamma(t) \ge -\epsilon.
\end{equation}
By standard Input-to-State Safe CBF results~\cite{kolathaya2018input}, subjecting a valid HOCBF dynamic to a strictly bounded perturbation $-\epsilon$ guarantees robust forward invariance of the deformed safe set $\mathcal{C}_{\mathrm{ISS}} = \{x : h(x) \ge -\delta(\epsilon)\}$.
\end{proof}

\section{Engineering Design Procedure and Real-Time Execution}
\label{sec:engin-design-proc}

With the theoretical bridge between the spatial geometry and temporal
filter established, the Space-Time Transform translates into a
deterministic, six-step design procedure for real-time control hardware.
By replacing improper differentiation with a causal LTI dynamic
extension, this architecture eliminates heuristic gain tuning and
provides a direct pipeline from hardware characterization to
discrete-time implementation.

\begin{enumerate}
    \item \textbf{Hardware Characterization:} Identify the physical limits of the plant, specifically the actuator bandwidth $\omega_P$ and the strict saturation bounds $\mathcal{U} = [-u_{\max}, u_{\max}]$.
    \item \textbf{Nominal Spatial Geometry:} Define the desired physical safety boundary $h(x)$ and determine its relative degree $n$ with respect to the control input.
    \item \textbf{Temporal Filter Design ($F(s)$):} Select a proper
      continuous-time low-pass filter $F(s)$ adhering to the
      restrictions in Section~\ref{sec:robust-forward-invariance}.
    \item \textbf{Analytic Spatial Synthesis ($C_k$):} Extract the
      characteristic polynomial from the denominator of $F(s)$.
      Following the horizon inversion pipeline
      (Section~\ref{sec:synth-spat-decay}), expand this polynomial to
      analytically extract the spatial coefficients $C_k$, fully
      defining the linear spatial extension
      $\alpha_L\big(\mathbf{h}(x)\big)$.
    \item \textbf{Dynamic State Extension:} Realize the temporal filter as a continuous LTI state-space system $(A_f, B_f, C_f)$. At each time step, feed the instantaneous, purely spatial Lie derivatives into the filter dynamics:
    \begin{equation}
        \dot{z}(t) = A_f z(t) + B_f \Big[ \mathcal{L}_f^n h(x(t)) + \alpha_L\big(\mathbf{h}(x(t))\big) \Big]
    \end{equation}
    The filter output rigorously evaluates the causal Space-Time integral: $\Omega_{\text{STT}}[h](t) = C_f z(t)$.
    \item \textbf{Discrete-Time QP Execution:} Discretize the LTI filter dynamics using a Zero-Order Hold (ZOH) for the hardware sample time $T_s$, yielding the discrete matrices $(A_d, B_d)$. At each discrete timestep $k$, update the filter state $z_{k+1} = A_d z_k + B_d w_k$, and solve the exact bounded Quadratic Program using the instantaneous control coefficient $b(x_k) = \mathcal{L}_g \mathcal{L}_f^{n-1}h(x_k)$:
    \begin{equation}
        \min_{u_k \in \mathcal{U}} \quad \frac{1}{2} \|u_k - u_{\text{nom}}\|^2
    \end{equation}
    \begin{equation*}
        \text{s.t.} \quad C_f z_k + b(x_k)u_k \ge 0
    \end{equation*}
\end{enumerate}

\section{Experimental Validation}
\label{sec:exper-valid}

To validate the theoretical guarantees of the Space-Time Transform (STT)
under high-frequency measurement noise, a third-order ($n=3$)
triple-integrator physical system $\dddot{x} = u$ was simulated under
symmetric actuator saturation bounds $u \in [-u_{\max}, u_{\max}]$ with
$u_{\max} = 3.2\text{ m/s}^3$. The plant is initialized in an aggressive
approach trajectory toward a hard spatial safety boundary at
$h(x) = x \ge 0$, with initial state
$(x_0, v_0, a_0) = (9.3\text{ m}, -6.0\text{ m/s}, 0.0\text{ m/s}^2)$
and a nominal baseline command $u_{\text{nom}} = -3.2\text{ m/s}^3$.

To stress-test spatial differentiation under realistic sensor regimes,
zero-mean Gaussian noise was injected into state feedback:
$\sigma_x = 0.1\text{ m}$ for position and $\sigma_v = 0.3\text{ m/s}$ for
velocity. A tight stochastic safety buffer
$\epsilon_{\text{robust}} = k_\sigma \sigma_x = 0.05\text{ m}$
($k_\sigma = 0.5$) was applied to evaluate performance near the boundary
edge. All architectures were executed at discrete steps of
$\Delta t = 0.005\text{ s}$ ($200\text{ Hz}$) using exact Zero-Order Hold
(ZOH) discretization over a $6.5\text{ s}$ horizon across $N = 50$ Monte
Carlo simulation runs.

\subsection{Evaluated Safety Filter Architectures}
\label{sec:eval-architectures}

Six safety filter configurations across four structural paradigms were benchmarked under identical initial conditions and noise sequences:

\begin{enumerate}
    \item \textbf{Nominal High-Order CBF (HOCBF):} Formulated with distinct spatial decay poles $a_1 = 1.3$, $a_2 = 1.5$, and $a_3 = 1.7$, yielding linear spatial gains $C_0 = a_1 a_2 a_3 = 3.315$, $C_1 = a_1 a_2 + a_1 a_3 + a_2 a_3 = 6.71$, and $C_2 = a_1 + a_2 + a_3 = 4.5$. Operating without a robust margin, the unsafe drift boundary evaluates to:
    \begin{equation}
        \alpha_{\text{HOCBF}}(x, v, a) = C_0 (x + n_x) + C_1 (v + n_v) + C_2 a
    \end{equation}
    The saturated jerk command is generated via $u(t) = \text{sat}\left(\max(-\alpha_{\text{HOCBF}}, u_{\text{nom}}), u_{\max}\right)$.

    \item \textbf{Nominal Exponential CBF (ECBF):} Parameterized by coincident spatial decay poles at $\omega_{\text{spatial}} = 1.5\text{ rad/s}$ and evaluated without a robust margin, yielding drift function $\alpha_{\text{ECBF}}(x, v, a) = C_0 (x + n_x) + C_1 (v + n_v) + C_2 a$.

    \item \textbf{Parameterized Barrier Functions (PBF):} Based on Alan
      et al.~\cite{alan2023parameterized,alan2025generalizing}, this
      architecture shifts the nominal HOCBF boundary by a parameterized
      spatial buffer $h^\star = \epsilon_{\text{robust}}$ to provide Input-to-State
      Safety against disturbances. The drift function becomes
      $\alpha_{\text{PBF}}(x, v, a) = C_0 (x + n_x - h^\star) + C_1 (v + n_v) +
      C_2 a$.

    \item \textbf{External Pre-Filtering:} Measurement feedback is passed through an external low-pass filter ($\omega_{\text{filter}} = 14.0\text{ rad/s}$) before differentiation:
    \begin{equation*}
        \dot{\tilde{h}}(t) = \omega_{\text{filter}} \left( (x(t) + n_x - \epsilon_{\text{robust}}) - \tilde{h}(t) \right)
    \end{equation*}

    \item \textbf{STT - Real Poles Topology:} Integrates spatio-temporal dynamics using a 3rd-order coincident real-pole filter kernel ($\omega_f = 14.0\text{ rad/s}$) embedded directly into the safety operator:
    \begin{align*}
      \dot{z}_R &= A_R z_R + B_R \alpha_{\text{Robust\_ECBF}}, \\
      u_{\text{STT-R}} &= \text{sat}\left(\max(-C_R z_R, u_{\text{nom}}), u_{\max}\right)
    \end{align*}
    where $A_R = \begin{bmatrix} 0 & 1 & 0 \\ 0 & 0 & 1 \\ -\omega_f^3 & -3\omega_f^2 & -3\omega_f \end{bmatrix}$, $B_R = \begin{bmatrix} 0 \\ 0 \\ \omega_f^3 \end{bmatrix}$, and $C_R = \begin{bmatrix} 1 & 0 & 0 \end{bmatrix}$.

    \item \textbf{STT - Bessel Topology:} Implements a 3rd-order Bessel filter kernel ($\omega_f = 14.0\text{ rad/s}$) designed for linear phase response, using matrices $A_{\text{BS}} = \begin{bmatrix} 0 & 1 & 0 \\ 0 & 0 & 1 \\ -15\omega_f^3 & -15\omega_f^2 & -6\omega_f \end{bmatrix}$ and $B_{\text{BS}} = \begin{bmatrix} 0 \\ 0 \\ 15\omega_f^3 \end{bmatrix}$.
\end{enumerate}

\begin{table*}[t!]
  \centering
  \caption{Monte Carlo Performance Comparison ($N=50$ Runs, $\sigma_x=0.1\text{ m}$, $\sigma_v=0.3\text{ m/s}$, $\epsilon_{\text{robust}}=0.05\text{ m}$)}
  \label{tab:mc_results}
  \begin{tabular}{l c c c}
    \hline
    \textbf{Controller Architecture} & \textbf{Safety Rate (\%)} & \textbf{Min. Margin $h_{\min}$ [m]} & \textbf{Control Chatter $\text{TV}(u)$} \\
    \hline
    Nominal HOCBF & $64.0\%$ & $0.01 \pm 0.02$ & $2186.1 \pm 42.6$ \\
    Nominal ECBF & $64.0\%$ & $0.01 \pm 0.02$ & $2196.7 \pm 42.8$ \\
    PBF (Alan et al.) & $100.0\%$ & $0.06 \pm 0.02$ & $2178.3 \pm 42.6$ \\
    Thm Part 1 (Ext. Filter) & $0.0\%$ & $-176.14 \pm 0.00$ & $0.0 \pm 0.0$ \\
    \textbf{STT (Real Poles)} & $\mathbf{100.0\%}$ & $\mathbf{0.06 \pm 0.02}$ & $\mathbf{20.9 \pm 1.6}$ \\
    STT (Bessel) & $92.0\%$ & $0.03 \pm 0.02$ & $46.7 \pm 2.9$ \\
    \hline
  \end{tabular}
  \label{tab:1}
\end{table*}

\begin{figure}[t!]
    \centering
    \includegraphics[width=\columnwidth]{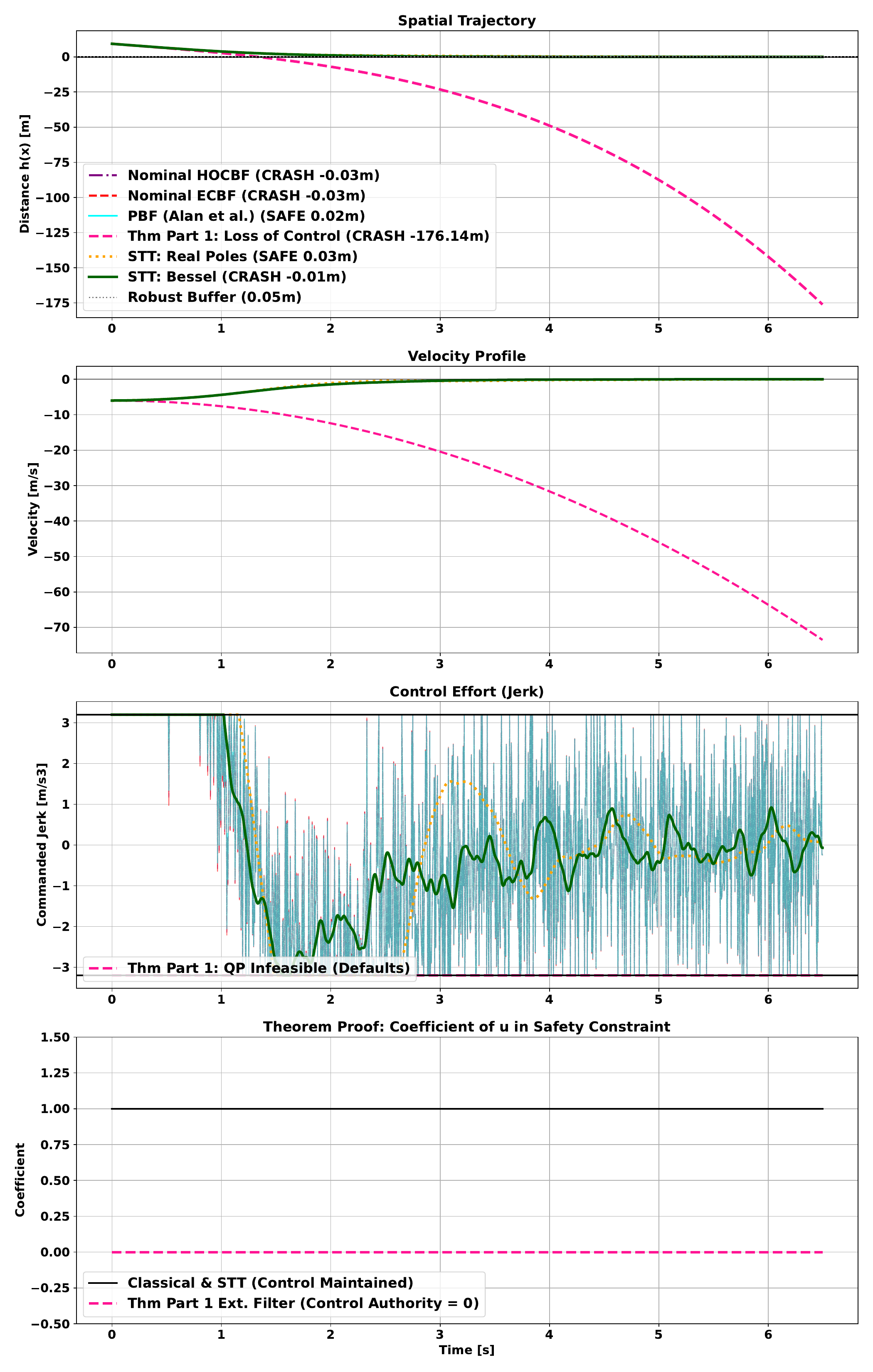}
    \caption{Single-run comparative trajectories for $n=3$ under measurement noise ($\sigma_x = 0.1\text{ m}, \sigma_v = 0.3\text{ m/s}$). \textbf{Row 1:} Spatial trajectories $h(x)$. \textbf{Row 2:} Velocity profiles. \textbf{Row 3:} Commanded jerk $u(t)$ demonstrating extreme chattering in standard and PBF CBFs versus smooth control in STT. \textbf{Row 4:} Affine control coefficient $\frac{\partial C(t)}{\partial u}$, confirming complete control authority loss ($\mathcal{L}_g C(t) = 0$) under external pre-filtering.}
    \label{fig:results}
\end{figure}

\subsection{Validation of Theorem Part 1: Complete Loss of Control Authority}
\label{sec:thm1-validation}

The experimental trajectories in Figure~\ref{fig:results} and the
statistical results in Table~\ref{tab:mc_results} provide direct
empirical confirmation of Theorem~\ref{thm:extfilter} Part 1. Passing
measurement signals through an external pre-filter before evaluating the
higher-order safety derivative reduces the affine control coefficient
$\mathcal{L}_g C(t) = \frac{\partial \tilde{h}^{(3)}}{\partial u}$ to zero throughout the entire
time horizon.

Because $u(t)$ is entirely removed from the active safety inequality constraint, the safety filter loses all control authority over the physical plant. The QP solver remains permanently infeasible and falls back to the nominal forward acceleration command ($u = u_{\text{nom}} = -3.2\text{ m/s}^3$). Consequently, as recorded in Table~\ref{tab:mc_results}, external pre-filtering achieves a \textbf{$0.0\%$ safety rate}, with zero control chattering ($\text{TV}(u) = 0.0 \pm 0.0$) and a catastrophic mean boundary penetration depth of \textbf{$-176.14\text{ m}$}.

\subsection{Validation of Theorem~\ref{thm:chatter}: Chatter
  Suppression and Robotic Hardware Implications}
\label{sec:classical-cbf-failure}

Standard memoryless formulations (Nominal HOCBF and ECBF) without robust
buffers suffer a low $64.0\%$ safety rate under raw feedback. Adding a
spatial buffer, as proposed by
PBF~\cite{alan2023parameterized,alan2025generalizing}, restores the
theoretical safety margin, successfully yielding a $100.0\%$ safety
rate. However, while PBF preserves spatial safety and active control
authority ($\mathcal{L}_g C(t) = 1.0$), it completely fails to attenuate temporal
noise.

High-frequency sensor noise injected into position and velocity is
repeatedly differentiated when evaluating $h^{(3)}(x)$, causing the
commanded jerk to chatter violently between the saturation bounds
$[-u_{\max}, +u_{\max}]$ (Row 3 of Table~\ref{tab:1}). Over the 50 Monte
Carlo runs, PBF exhibits an extremely high Total Variation metric of
$\text{TV}(u) = 2178.3 \pm 42.6$. In physical robotic platforms, such as
autonomous ground vehicles and quadrupedal robots, this high-frequency
control chattering poses severe operational hazards, including gearbox
fatigue, thermal dissipation, and structural resonance.

By contrast, the STT (Real Poles) architecture resolves this fundamental
trade-off. As proven in Theorems~\ref{thm:extfilter}
and~\ref{thm:stt_forward_invariance}, embedding spatial dynamics into a
spatio-temporal kernel filters high-frequency noise while maintaining
control authority ($\mathcal{L}_g C(t) = 1.0$). STT (Real Poles) perfectly matches
the $100.0\%$ safety rate of PBF while simultaneously reducing control
chatter by over $99\%$, dropping $\text{TV}(u)$ down to $20.9 \pm 1.6$.

\subsection{Filter Pole Topology, Phase Lag, and Transient Overshoot}
\label{sec:filter-topology}

Operating near the physical safety limit with a calibrated robust buffer ($\epsilon_{\text{robust}} = 0.05\text{ m}$) highlights the critical role of kernel pole selection in STT synthesis:
\begin{itemize}
\item \textbf{STT Real Poles ($100.0\%$ Safety Rate):} Critically damped
  real poles introduce minimal phase delay and zero transient overshoot
  in step response, allowing the controller to decelerate smoothly and
  consistently halt prior to boundary penetration
  ($h_{\min} = 0.06 \pm 0.02\text{ m}$).
\item \textbf{STT Bessel ($92.0\%$ Safety Rate):} Although optimized for
  linear phase delay, the higher initial group delay of the 3rd-order
  Bessel filter delays the onset of maximum braking effort. This dynamic phase lag leads to 
  unwanted transient overshoot, triggering physical boundary violations
  in $8\%$ of runs and increasing control chatter ($\text{TV}(u) = 46.7 \pm 2.9$).
\end{itemize}

These findings confirm that dynamic spatio-temporal kernel design via real-pole STT is essential to achieve noise suppression without introducing destructive phase lag or sacrificing physical safety.

\section{Related Work}
\label{sec:related_work}

The mathematical formalization of safety in control systems has been largely driven by the development of Control Barrier Functions (CBFs). The foundation of modern CBFs provides a geometric framework that unifies safety and performance through Control Lyapunov-Barrier Function Quadratic Programs (CLF-CBF QPs) \cite{ames2017control}. To handle systems where the control input does not appear in the first derivative of the safety constraint, researchers developed relative-degree barriers, prominently including High-Order CBFs (HOCBF) \cite{xiao2019control} and Exponential CBFs (ECBF) \cite{nguyen2016exponential}. These variants, alongside reciprocal barrier formulations \cite{ames2017control}, established the dominant paradigm of enforcing spatial boundaries via affine QP formulations.

As the field matured, the literature expanded to address various
practical implementation challenges. Significant efforts have been
directed toward Sampled-data CBFs to handle discrete-time
implementation~\cite{taylor2022safety}, Robust CBFs to handle bounded
model uncertainties~\cite{alan2023parameterized, alan2025generalizing},
Predictive CBFs to mitigate myopia~\cite{breeden2022predictive}, and
Delayed CBFs for systems with actuator or state
delays~\cite{singletary2020control}. However, our fundamental Theorem
differs from these approaches because it attacks the underlying
mathematical kernel structure of the barrier function itself, rather
than addressing external uncertainty, sample-and-hold effects, or
predictive horizons.

Dynamic Control Barrier Functions (D-CBFs) \cite{jian2023dynamic}
introduce auxiliary dynamic internal states to accommodate higher
relative degrees and input constraints. Although D-CBFs introduce
auxiliary dynamics, the safety constraint is still synthesized from
instantaneous Lie differentiation rather than a proper spatio-temporal
kernel. Similarly, Input-to-State Safe CBFs (ISS-CBFs)
\cite{kolathaya2018input} counteract measurement noise and external
perturbations by expanding conservative spatial safety buffers via
class-$\mathcal{K}$ disturbance margins, yet they preserve the underlying
improper, memoryless derivative operator.In contrast, the proposed STT
unifies spatial boundary evaluation with temporal dynamic filtering by
synthesizing a proper spatio-temporal kernel
$G_{\text{STT}}(s) = F(s) G_{\text{improper}}(s)$. Rather than relying
on conservative set contractions or uncompensated dynamic extensions,
STT directly bounds the demanded control effort by anchoring spatial
decay to the temporal filter bandwidth. This architecture suppresses
high-frequency chatter while natively guaranteeing exact physical
forward invariance and preserving full affine control authority.

While filtering techniques have been introduced in the CBF literature,
they generally treat the filter as an external augment to a
fundamentally memoryless barrier mechanism. For instance, approaches
targeting smooth control often apply low-pass filtering externally to
the synthesized downstream control signal, effectively smoothing the
output of the CBF QP~\cite{liu2025ensuring}. Similarly, stochastic and
observer-based CBF formulations rely on upstream filtering, such as
Extended Kalman Filters (EKFs), to smooth noisy measurements into state
estimates prior to evaluating a traditional
CBF~\cite{clark2021automatica}. However, as established in
Theorem~\ref{thm:extfilter}, this methodology is suspect for physical
control systems. Recent measurement-robust frameworks attempt to
mitigate state estimation error via learning-based
adjustments~\cite{rober2026learning}, yet they still fundamentally rely
on passing corrected measurements to a zero-memory spatial operator. In
stark contrast, STT integrates the temporal filter directly into the
structural formulation of the safety constraint. By synthesizing a
proper spatio-temporal kernel, STT resolves the improper high-frequency
amplification at its root, natively preserving affine control authority
without falling into the surrogate boundary trap of external observers.

Since the proposed STT maps spatial geometry into temporal filter states
via continuous-time convolution, our dynamic extension naturally
connects to the broader literature on Volterra integral
operators~\cite{brunner1945volterra}, distributed-delay systems, and
hereditary control systems~\cite{delfour1972controllability} involving
memory operators and functional differential
equations~\cite{hino2006functional}. However, we do not require a
lengthy analysis of infinite-dimensional state histories; we simply
acknowledge these mathematical connections while emphasizing that our
STT contribution is fundamentally a synthesized safety operator, not
merely another generic hereditary system.

Recent investigations have increasingly exposed the structural
vulnerabilities of naive CBF safety filters. While Choi et
al.~\cite{choi2026safety} demonstrate that CBF safety filters can induce
internal state divergence due to non-minimum phase zero dynamics under
non-negative barrier inputs, this paper addresses the dual external
failure mode: control demand divergence. We show that even when internal
dynamics are stable, the improper transfer function of classical spatial
differentiation forces unbounded high-frequency control effort, making
continuous QP feasibility impossible under real actuator limits.

\section{Conclusion and Future Work}
\label{sec:concl-future-work}

By exposing the HOCBF and ECBF frameworks as improper zero-memory
filters, we diagnose the root cause of high relative-degree QP
infeasibility. The Space-Time Transform cures this by enforcing safety
through proper temporal convolution, guaranteeing rigorous mathematical
safety bounds that respect the physical realities of real hardware.

Future work will focus on the exhaustive functional analysis of the
Space-Time Transform operator space. While the foundational algebraic
properties, such as linearity, causal representation, and invertibility,
have been established here to enable physical hardware synthesis, a
rigorous functional analysis is required to fully characterize the
theoretical boundaries of spatio-temporal safety filters across broader
classes of control systems.

\bibliographystyle{IEEEtran}
\bibliography{ref.bib}

\end{document}